\documentclass[11pt,a4paper]{article}
\usepackage{authblk}

\usepackage{graphics}
\usepackage{graphicx}
\usepackage{amssymb,amsmath,amstext,amsfonts}

\usepackage{a4wide}

\usepackage{fullpage}
\usepackage{amsfonts,amsmath,amssymb,amsthm,boxedminipage,color,url,fullpage}
\usepackage{enumerate,paralist}
\usepackage[numbers]{natbib} % \citet{foo} -> Foo et al. [5]
\usepackage[labelfont=bf]{caption}
\usepackage{aliascnt,cleveref}

\usepackage{epic,eepic}

\newtheorem{lemma}{Lemma}[section]
\newtheorem{proposition}[lemma]{Proposition}
\newtheorem{conjecture}[lemma]{Conjecture}
\newtheorem{theorem}[lemma]{Theorem}

\theoremstyle{definition}

\theoremstyle{remark}

\providecommand{\keywords}[1]{\textbf{\text{Keywords:}} #1}

\begin{document}

\title{Multicast Network Design Game on a Ring}

\author[1]{Akaki Mamageishvili}
\author[2]{Mat\'u\v{s} Mihal\'ak}

\affil[1]{Department of Computer Science, ETH Zurich, Switzerland}
\affil[2]{Department of Knowledge Engineering, Maastricht University, The Netherlands}
\date{July, 2015}

\maketitle

\begin{abstract}

In this paper we study quality measures of different solution concepts for the multicast network design game on a ring topology.
We recall from the literature a lower bound of $\frac{4}{3}$ and prove a matching upper bound for the price of stability, which is the ratio of the social costs of a best Nash equilibrium and of a general optimum. Therefore, we answer an open question posed by Fanelli et al. in \cite{ring-design-journalversion}. 
We prove an upper bound of $2$ for the ratio of the costs of a potential optimizer and of an optimum, provide a construction of a lower bound, and give a computer-assisted argument that it reaches $2$ for any precision. 
We then turn our attention to players arriving one by one and playing myopically their best response. We provide matching lower and upper bounds of $2$ for the myopic sequential price of anarchy (achieved for a worst-case order of the arrival of the players).
We then initiate the study of myopic sequential price of stability and for the multicast game on the ring we construct a lower bound of $\frac{4}{3}$, and provide an upper bound of $\frac{26}{19}$. 
To the end, we conjecture and argue that the right answer is $\frac{4}{3}$.
%
%We focus on equilibria that minimize the underlying potential function of the
%game, and observe and exploit their properties to obtain the new upper bound.
%

  %, fair cost-allocation}
\end{abstract}

%\keywords{Network design game, Nash equilibrium, Price of Stability}
%\begin{keywords} Network design game, Nash equilibrium, Price of Stability \end{keywords}

%%%%%%%%%%%%%%%%%%%%%%%%%%%%%%%%%%%%%%%%%%%%%%%%%%%%%%%%%%%%%%%%%%%%%%%%%%%%%%%%
%%%   Introduction
%%%%%%%%%%%%%%%%%%%%%%%%%%%%%%%%%%%%%%%%%%%%%%%%%%%%%%%%%%%%%%%%%%%%%%%%%%%%%%%%
%
\section{Introduction}

Network design game is played by $n$ players on an edge-weighted graph. Each player $i$, $i=0,\ldots,n-1$, connects her terminal vertices $s_i$ and $t_i$ by selecting an $s_i$-$t_i$ path $P_i$. Using an edge $e$ costs $c_e$ and all players using it share the cost equally. 
In total, player $i$'s cost for using path $P_i$ is the sum of all shares towards the edges of $P_i$.

Network design game belongs to the broader class of congestion games. It is a special congestion game in that increasing the congestion on a resource makes it cheaper to use (in contrast to the more established and studied games with monotone increasing cost functions).
Finite congestion games are exact potential games, i.e., games for which a potential function exists, i.e., a function $\Phi(P_0,\ldots,P_{n-1}) \rightarrow \mathbb{R}$ that exactly reflects the difference in any player's cost, if this unilaterally changes her path from $P_i$ to $P_i'$.
It is well-known that exact potential games always possess a pure Nash equilibrium, for example the vector $(P_0,\ldots,P_{n-1})$ minimizing the potential function $\Phi$.
The price of anarchy is the ratio of the worst Nash equilibrium cost and the general optimum cost, and can be as large as $n$. The price of stability, which is the ratio of the best Nash equilibrium cost and the general optimum cost, of network design games is well understood for directed graphs -- it is at most $H_n$ \cite{original}, where $H_n = 1 + 1/2 + 1/3 + \cdots + 1/n$ is the $n$-th harmonic number (equal asymptotically to $\log n$) and the matching lower bound example has also been constructed.
The price of stability of the game is much less understood for undirected graphs. While it is known to be strictly smaller than $H_n$~\cite{Matus,mfcs2014}, namely, at most $H_{n/2}$~\cite{mfcs2014}, the largest known example has price of stability equal to roughly 2.245~\cite{lowerbounds}.
Closing this gap is a major open problem in the field of congestion games and in the computational game theory in general.

For the special type of the game where all players have the same target vertex $t$, better bound on the price of stability has been proven~\cite{multicast}. If additionally each vertex of a graph is a source vertex of some player, a series of papers improved the upper bound ~\cite{Fiat+etal/2006,Lee+Ligett/2013,constant}, where the latest result of Bil\`o et al.\,\cite{constant} shows that price of stability is $O(1)$ in this class of games.
In many results, the potential function, and the equilibria minimizing it, play an important role. 
Actually, equilibria minimizing potential function are regarded as stable (against noise) by Asadpour and Saberi~\cite{Asadpour+Saberi/2009}, and accordingly, some authors studied the price of stability restricted to these kind of equilibria~\cite{popos,mfcs2014}, the so-called \emph{potential-optimum price of stability}. 

One of the motivations to study best Nash equilibria is that they can be regarded as outcomes of the game if a little coordination is present -- an authority that suggests the players the strategies $P_i$. Then, players have no incentive to unilaterally deviate from the suggested strategy profile. 
It is questionable whether such an authority exists -- it would need to be very strong, both computationally and imperatively. 
To address this applicability issue of equilibrium concepts, sequential versions of the game were studied: the players arrive one by one, and upon arrival, player $i$ chooses \emph{myopically} the best path $P_i$ as if this was the end of the game (i.e., no further players would arrive). Chekuri et al.\,\cite{sequentialOriginal} show that the total cost achieved by a worst-case permutation of the arriving players is at most $O(\sqrt{n}\log n)$ times the optimum cost. 
Subsequently, Charikar et al.\,\cite{sequentialImproved} improved this bound to $O(\log^2 n)$ (the original version~\cite{sequentialImproved} is erroneous, but the authors provide corrected arguments upon request).  
The worst-case approach to the order in which the players arrive naturally models the complete lack of coordination. 
In this paper, we suggest to study also the best-case order in which players arrive. This is motivated by the presence of an authority that can control the access to the resources over time (and thus decide an order of the arriving players). Such an authority is arguably weaker than the one mentioned above, as it does not impose any decision upon the players, and it leaves them to decide their strategies freely upon arriving. 
Bil\'{o} et al. \cite{graphical-cost-sharing} studied a version of a cost sharing scheme for multicast network design game, in which each player only knows strategies of some other players, and pays fair share of edge costs that she uses based only on her information. Sequential versions described above can be modeled with this cost sharing scheme. 

In this paper, we focus on one specific network topology: the ring. This is a fundamental topology in networking and communications. It is the edge-minimal topology that is resistant against a single link fault. From the decentralized point of view, \emph{call control} comes close in spirit to network design games, in that the connecting $s_i$-$t_i$ paths needs to be chosen to obey given capacities on the links~\cite{CallControl}. The study of approximation algorithms is the counterpart to bounding the prices of anarchy and stability. Rings have also been intensively studied in the distributed setting, e.g., among plenty of others, in the context of the fundamental leader election problem \cite{computingLeaderonrings}.

Network design games on rings has previously been studied by Fanelli et al.\,\cite{ring-design-journalversion}, which show a tight bound of 3/2 on the price of stability for the general setting.
In this paper we restrict ourselves to the \emph{multicast} version in which all players share the same target vertex $t=t_i$, $i=0,\ldots,n-1$ and answer the open question asked by Fanelli et al.\,\cite{ring-design-journalversion} about tight bounds of the price of stability for multicast game on a ring. 
We study various solutions concepts and analyze their quality compared to an optimum network (with respect to the social cost). In most cases, we are able to provide tight bounds. Furthermore, we also study the myopic sequential price of stability in general multicast network design games, and give a simpler proof of an upper bound of $4$ for this class of games compared to a more general proof in \cite{graphical-cost-sharing} (cf. this with the upper bound of $\log^2 n$ on the myopic sequential price of anarchy for multicast games).

\section{Preliminaries}

\emph{Network design game} is a strategic game of $n$ players played on
an edge-weighted graph $G=(V,E)$ with non-negative edge costs $c_e$, $e\in E$.
Each player $i$, $i=0,\ldots,n-1$, has a dedicated \emph{source} node $s_i$ and a \emph{target} node $t_i$. In the multicast game all $t_i$'s are the same and we denote it by $t$ throughout a paper. 
All $s_i$-$t_i$ paths form the set $\mathcal{P}_i$ of the \emph{strategies} of player $i$. Each player $i$ chooses one path $p_i\in \mathcal{P}_i$, the union of which creates a network in which all $s_i$-$t_i$ pairs are connected.
The cost for player $i$ is $\displaystyle \sum_{e\in p_i}\frac{c(e)}{n(e)}$ where $n(e)$ is the number of players that use the edge $e$ in their chosen paths.
A strategy profile $p$ is a vector of strategies for all players, $p=(p_0,\ldots, p_{n-1})$. A strategy profile is Nash equilibrium if no player can unilaterally change her strategy $p_i$ to $p_i'$ and improve her cost. 
The (social) cost of a strategy profile $p$ is the cost of the created network, i.e., the sum of the costs of all edges in the created network, which is, in turn, the sum of all players' costs.
An optimum network is a network of minimum social cost in which all $s_i$-$t_i$ pairs are connected. An optimum network can be equivalently described by a strategy profile $p^*$, and then we refer to $p^*$ as an optimum strategy profile.
Note that an optimum strategy profile is not, in general, a Nash equilibrium. 
Observe also that in a multicast game an optimum network forms a Steiner tree on the terminals $s_i$ and $t$ for $i=0,\ldots,n-1$.
If an underlying graph $G$ is a ring, then there are only $2$ possible strategies for each player. 

In this paper, we focus on the multicast game on rings. We can assume, without loss of generality, that every node but the target $t$ is a source of exactly one player.
Otherwise, we can modify the topology by the following two operations.
If there are $l>1$ players sharing the same node $x$ of the ring as a source vertex, we make $l$ copies of this vertex, add $l-1$ consecutive edges of cost $0$ between them to make a path of length $l-1$, replace $x$ in the ring with this path in a natural way, and associate each vertex with a unique source (copy of $x$).
If there is a node $x$ in the ring which is not a target nor a source of any player, we delete $x$ from the ring, and connect its two neighbors by an edge of cost $c_e+c_{e'}$, where $e,e'$ are the two adjacent edges of $x$.
A repetitive application of these two operations preserve the cost of optimum and Nash equilibrium strategy profiles, and also preserves the equilibrium properties of strategy profiles (if the strategies are expressed in the form ``go clockwise/counterclockwise to $s_i$").
\begin{figure}[t]
  \centering
  \includegraphics[scale=1]{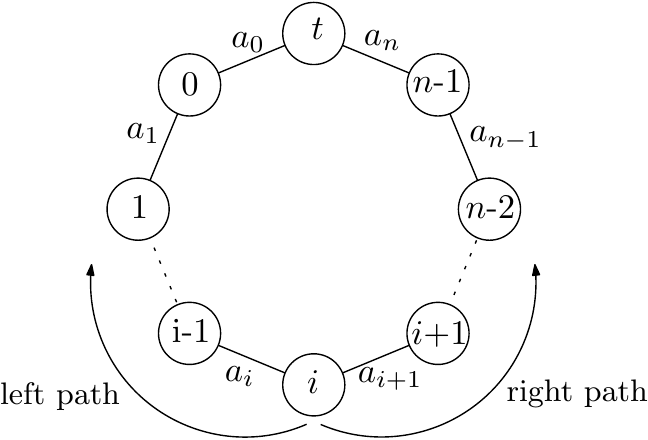}
  \caption{Multicast game on rings.}
  \label{Example of a game}
\end{figure}

We label the sources (players) and the edges connecting them in a counter-clockwise order as in Fig. \ref{Example of a game}, where $a_i$ denotes the cost of the $i$-th edge. Player $i$ has exactly $2$ strategies, one is to go \emph{left}, i.e., \emph{clockwise}, taking edges $i,i-1,\ldots 0$, or to go \emph{right}, i.e., \emph{counterclockwise}, taking edges $i+1,\ldots n$. 
Observe that the optimum strategy profile is the one which uses all edges except the most expensive edge. Let $o$ denote the most expensive edge. Then the (social) cost of an optimum network is $\sum_{i\neq o}a_i$. 

\emph{Price of anarchy} of a network design game is the ratio of the costs of a worst Nash equilibrium and of an optimum strategy profile.
\emph{Price of stability} is the ratio of the costs of a best Nash equilibrium and of an optimum strategy profile.  
\emph{Potential optimum} is a strategy profile $p$ that minimizes the potential function $\Phi=\sum_e \sum_{i=1}^{n_e} \frac{c_e}{i}$.
\emph{Potential-optimum price of anarchy/stability} is the ratio of the costs of the worst/best potential-optimum profile and of an optimum strategy profile.
\emph{Myopic sequential price of anarchy/stability} is the worst-case/best-case ratio of the costs of a strategy profile that can be obtained by ordering the players as in a permutation $\pi$ and letting player $\pi(i)$ choose the best-response $p_{\pi(i)}$ in the game induced by the first $i$ players $\pi(1),\pi(2),\ldots,\pi(i)$ and of an optimum profile.

\paragraph{Note on related concepts.} The term \emph{sequential price of anarchy} has been used \cite{sequentialPoA,sequentialIsolation} to express a different, yet still closely related, concept compared to the notion of the myopic sequential price of anarchy/stability. 
In the sequential price of anarchy, players also come one by one, and decide their strategy upon arrival, but the stability of the outcome is measured in terms of Nash equilibria again. In some sense, the game resembles extensive games. Observe that profiles $p$ that get compared to optima in the myopic sequential price of anarchy/stability are in general no Nash equilibria.

%%%%%%%%%%%%%%%%%%%%%%%%%%%%%%%%%%%%%%%%%%%%%%%%%%%%%%%%%%%%%%%%%%%%%%%%%%%%%%%%
%%%   Original Price of stability 
%%%%%%%%%%%%%%%%%%%%%%%%%%%%%%%%%%%%%%%%%%%%%%%%%%%%%%%%%%%%%%%%%%%%%%%%%%%%%%%%

\section{Price of Anarchy/Stability for Multicast on Rings}

It is known that the price of anarchy on general graphs is at most $n$, and that this bound is tight. The tight example actually is a multicast game on a ring, and the general analysis of the price of anarchy thus carries over to our multicast game on rings. For completeness, we show the example in Fig.~\ref{fig:price_of_anarchy}.

\begin{figure}[t]
\begin{minipage}[t]{0.5\textwidth}
%\begin{figure}
	\centering
	\includegraphics[scale=1]{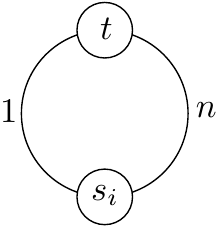}
	\caption{In the worst equilibrium, all players use the edge of cost $n$.}
	\label{fig:price_of_anarchy} 
%\end{figure}
\end{minipage}
\begin{minipage}[t]{0.5\textwidth}
%\begin{figure}[t]
  \begin{center}
  \includegraphics[scale=1]{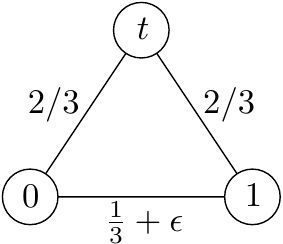}
  \caption{Example of a lower bound 4/3.}
  \label{4/3 lower bound}
  \end{center}
%\end{figure}
\end{minipage}
\end{figure}

\begin{theorem}[\cite{original}]	
	Price of anarchy for mutlicast games on rings is at most $n$. This is tight.
\end{theorem}

We now turn our attention to the price of stability. The example from Fig.~\ref{4/3 lower bound}, due to Anshelevich et al.\,\cite{original}, shows that the price of stability can be as high as 4/3 (observe that the game possesses a unique Nash equilibrium where both players use the direct edge to get connected to $t$).
We now show that the price of stability cannot get larger than that for multicast games on rings, and therefore answer the open question asked by Fanelli et al. \,\cite{ring-design-journalversion}.

%In this section we show that price of stability in the worst case is $\frac{4}{3}$. The multicast game depicted in the Fig. \ref{4/3 lower bound} has only one Nash equilibrium, a strategy profile that connects both players directly to the target vertex, whilst 
%optimum strategy has a cost $1+\epsilon$, meaning that the ratio is tending to $\frac{4}{3}$ when $\epsilon$ tends to $0$.
%
%Next we prove an upper bound of $\frac{4}{3}$. 

\begin{theorem}
\label{thm:pos<=4/3}
Price of stability in the multicast game on rings is at most $\frac{4}{3}$. 
\end{theorem}

In the proof of the theorem we will use the following lemma.

\begin{lemma}
	\label{lem:endpoints_of_optimum_edge_improving}
If a strategy profile $p$ in which an edge $i$ is not used is not Nash equilibrium, then either player $i$ or player $i-1$ can improve her cost by changing her strategy. 
\end{lemma}

\begin{proof}
\label {firstLemma}
Since the strategy profile $p$ is not a Nash equilibrium, there exists a player $k$ that can change her strategy and improve the cost. 
Assume, without loss of generality, that $k<i-1$. 
Since edge $i$ is not used in $p$, it follows that player $k$ uses the left path to get to $t$.
The cost of $k$ in $p$ is thus $\displaystyle \sum_{l=0}^{k}\frac{a_l}{i-l}$, which is, by our assumption, bigger than the cost of $k$ if she switches to the right path, i.e., bigger than $\displaystyle \sum_{l=k+1}^{i-1} \frac{a_l}{i-l+1} + \displaystyle \sum_{l=i}^{n}\frac{a_l}{l-i+1}$. It follows that player $i-1$ also uses the left path in $p$, and thus her cost is at least the cost of player $k$, whereas the alternative cost of $i-1$ if she switches to the right path is at most the alternative cost of player $k$. Hence, the alternative cost of player $i-1$ is smaller than her cost in $p$, and player $i-1$ thus improves her cost as well. 
%
%Assume that the claim is not true and without loss of generality a player $k<i-1$ can improve her strategy by changing. Note that player $k$'s current cost is $\displaystyle \sum_{l=0}^{l=k}\frac{a_l}{i-l+1}$ and by the assumption it is bigger than the alternative cost, which is equal to $\displaystyle \sum_{l=k+1}^{l=i-1}\frac{a_l}{i-l+1}+\displaystyle \sum_{l=i}^{l=n}\frac{a_l}{l-i+1}$, whilst cost of player $i-1$ in the current strategy is at least the cost of player $k$ her cost in the alternative strategy is at most the cost of player $k$ in the alternative strategy. This is a contradiction. 
\qed
\end{proof}

\begin{proof}[of Theorem~\ref{thm:pos<=4/3}]
Consider an optimum strategy profile and let $o$ be the edge that is not used in it. 
If optimum is also Nash equilibrium, then price of stability is 1 and the claim follows. 
Otherwise, optimum is not a Nash equilibrium and, by Lemma~\ref{lem:endpoints_of_optimum_edge_improving}, one of the endpoints of the edge $o$ can improve its cost. Assume, without loss of generality, that player $o-1$ can improve. We now consider the following best-response dynamics: let $o-1$ improve; then, edge $o-1$ is not used, and in case we have not reached Nash equilibrium, let player $o-2$ improve (the player $o-2$ must be able to improve by Lemma~\ref{lem:endpoints_of_optimum_edge_improving}), and so on, until some player $o-k$ cannot improve anymore (this happens at the latest for player 0), and we reach a Nash equilibrium. 

We will show that the social cost of a Nash equilibrium that is reached by this best response dynamics is maximized for $k=1$, i.e., for the strategy profile reached after one step of the dynamics. We then show that the cost of such a profile is at most $4/3$ times the cost of the optimum, which proves the theorem.

Let us first show the second part. Assume therefore that player $o-1$ switches to improve her cost, and the resulting profile is an equilibrium. In particular, we have that player $o-2$ does not want to switch.
This can be expressed by the following two inequalities: $\displaystyle \sum_{l=o}^{l=n}\frac{a_l}{l-o+1}\leq \displaystyle \sum_{l=0}^{l=o-1}\frac{a_l}{o-l}$, and $\displaystyle \sum_{l=0}^{l=o-2}\frac{a_l}{o-1-l}\leq \displaystyle \sum_{l=o-1}^{l=n}\frac{a_l}{l-o+2}$.
We further introduce a normalization of the edge costs so that the edges in the optimum sum up to 1. Thus, we obtain the \emph{normalization} equation $\displaystyle \sum_{i=0,i\neq o}^{n} a_i =1$. 
Now, taking the first inequality with weight 5, the second with weight 1, and the normalization equality with weight 6, we obtain that the cost of the Nash equilibrium where edge $o-1$ is not used has cost 
$\displaystyle \sum_{i=0, i\neq o-1}^{i=n}a_i$ at most $\frac{4}{3}$.

We can proceed in the same way for every other value of $k=2,3,\ldots$ for which the reached Nash equilibrium does not use edge $o-k$. For every $k$, we get for each of the players $o-k-1,o-k,\ldots,o-1$ an inequality stating that the player did not want, respectively wanted to swap her strategy. 
For all values of $k=1,2,3,4,5,6,7$, we provide in the appendix the coefficients with which we need to take the inequalities and to obtain the upper bound of at most 4/3 on the cost of the Nash equilibrium. 

If the length of the best-response dynamics is 8 or more, it follows that we do not need to add further inequalities, and the 7 inequalities obtained for the first 7 deviating players are enough to show the upper bound of 4/3 on the cost of the reached Nash equilibrium.
\qed 
\end{proof}

%From now on, $G$ is an arbitrary Shapley network design game with $n$ players,
%$O=(O_1,\ldots,O_n)$ is an arbitrary social optimum, and $N=(N_1,\ldots,N_n)$ is
%a Nash equilibrium minimizing the potential function.
%
%Since $N$ has minimal potential, we have $\Phi(N)\leq \Phi(S)$ for any strategy
%profile $S$, a fact that we will use in the following sections to prove Theorem
%\ref{mainthm}.

\section{Potential-Optimum Price of Anarchy for Multicast on Rings}

The potential-optimum price of anarchy/stability has been first studied, in the context of the network design games, by Kawase and Makino~\cite{popos}. Besides other results, they proved that for multicast network design games, the two values collide. Therefore, in the following, we only study the potential-optimum price of anarchy (POPaA for short), and we show that it is at most two for rings, and provide an infinite family of examples with increasing POPoA, which we conjecture converges to two, but leave the formal analysis as an open problem. We have analyzed one such game from the family which shows that POPoA can be as large as 1.99992. 
%

%Network design game is an exact potential game, and the general upper bound for the price of stability is obtained via potential functions. In \cite{popos} authors initiated a study of potential optimum price of anarchy and potential optimum price of stability. Namely, $POPoA = \frac{max_{P\in argmin\Phi}cost(P)}{min_{P\in\mathcal{P}}cost(P)}$ and $POPoS = \frac{min_{P\in argmin\Phi}cost(P)}{min_{P\in\mathcal{P}}cost(P)}$. They also proved that in case of multicast game they are equal. 
%We will prove that $POPoA$ for the multicast game on a ring is upper bounded by $2$ and construct a lower bound which is approaching $2$ by any precision that we tried, though we do not have a proof that the process is actually approaching $2$. 

\begin{theorem}
\label{potentialTheorem}
POPoA is at most $2$ in the multicast game on rings. 
\end{theorem}

\begin{proof}
Consider an optimal strategy profile $O$ and let $o$ be the edge that is not used in it. 
Consider a potential optimum strategy profile $P$ and let $p$ be the edge in it that is not used by any player. 
Assume, without loss of generality, that $p<o$.

By the definition of $P$, we have, for any strategy profile $Q$, $\Phi(P)\leq \Phi(Q)$, and in particular $\Phi(P)\leq \Phi(O)$, i.e., 

\begin{equation}
	\label{potentialInequality}
	\displaystyle \sum_{i=0}^{p-1} a_i\cdot H_{p-i}+\displaystyle \sum_{i=p+1}^{n}a_i\cdot H_{i-p} \leq  \displaystyle \sum_{i=0}^{o-1}a_i\cdot H_{o-i}+\displaystyle \sum_{i=o+1}^{n}a_i\cdot H_{i-o} \text{.}
\end{equation} 

We now concentrate on $a_o$ and show that $a_o$ is at most the cost of optimum, i.e., at most $\displaystyle \sum_{i\neq o}a_i$. This then shows that any strategy profile (and, in particular, $P$) has cost at most twice the cost of optimum.

Isolate in the second sum of the left hand side (LHS for short) of Equation~\eqref{potentialInequality} the term with $a_o$ and put the rest of the sum to the right hand side (RHS). This rest will dominate the second sum on the RHS, and by neglecting the resulting negative number, we get that
$ \displaystyle\sum_{i=0}^{p-1}a_i\cdot H_{p-i} + a_o\cdot H_{o-p} \leq \displaystyle\sum_{i=0}^{o-1} a_i\cdot H_{o-i} $, or, equivalently, that $a_o \leq \frac{\sum_{i=0}^{o-1} a_i\cdot H_{o-i} - \sum_{i=0}^{p-1}a_i\cdot H_{p-i}}{H_{o-p}}$.
Analyzing the influence of $p$ on the RHS, one can show that the RHS is maximized for $p=1$. Thus, we obtain that $a_o \leq \frac{\sum_{i=0}^{o-1} a_i\cdot H_{o-i}-a_0}{H_{o-1}}$. Then, since $H_o-1\leq H_{o-1}$ we get that $a_o \leq \sum_{i=0}^{o-1} a_i \leq \sum_{i\neq o} a_i$, which proves the claim and thus the theorem.
%
%
%We use this inequality together with the normalization equation $\displaystyle \sum_{i=0,  i\neq o}^{n}a_i = 1$ to maximize $\displaystyle \sum_{i=0 , i\neq p}^{n}a_i$.
%
%
%Since our goal is to upper bound $a_o-a_p$, then in the best case $o=n$, also $a_i=0$ for $i=p,\ldots, o-1$. If there are any non-zero cost edges $a_i$, with $i>o$, this would only increase a potential of a potential optimizer, and therefore give a less freedom to a variable $a_o$ to be larger. If there is any non-zero cost edges $a_i$ with $p<i<n$ then it would be better if the sum of these edge costs was added to an edge cost of $a_{p-1}$, and the edge costs become $0$. This way the potential of potential optimizer would decrease and the potential of an optimum would increase, again giving more freedom to variable $a_o$.
%
%Then the first inequality \ref{potentialInequality} has the form $\displaystyle \sum_{i=0}^{p-1}a_i\cdot H_{p-i}+a_o\cdot H_{n-p}\leq \displaystyle \sum_{i=0}^{p}a_i\cdot H_{n-i}$. In the best case $p=1$, $a_p = 0$, $a_0=1$ and then $a_o\leq \frac{H_n-1}{H_{n-1}}$, which is converging to $1$, as $n$ tends to infinity.  
%
%Certainly, the cost of $P$ is at most the cost of $O$ plus $a_O$ which is at most two (recall the normalization).
\qed
\end{proof}

We now provide a construction of a game which shows that POPoA is at least 1.99992. We conjecture that the construction can be used to prove an asymptotic lower bound of 2 on POPoA.

%Natural attempt to provide a tight example is to set $a_o:= \frac{\sum_{i=0}^{o-1} a_i\cdot H_{o-i}}{H_{o-1}}$.
%
%This, however, does not give a valid lower bound: the strategy profile where edge $0$ is not used has potential 
%
%
%Recall from the proof of Theorem \ref{potentialTheorem} that values $p=1$, $o=n$, $a_0=1$, $a_1=0$ and $a_o=\frac{H_n-1}{H_{n-1}}$ give the worst case bound for POPoA. Yet, this construction does not give a valid lower bound: a strategy profile in which edge $0$ is not used has potential $H_n\cdot \frac{H_n-1}{H_{n-1}}$, which is smaller than $H_n$, a potential of the potential optimizer.
%
%We can resolve this problem with the following approach.
%
Consider non-zero numbers $a_0,\ldots,a_{2\cdot l}$ that sum up to $1$, and where $l$ is constant, $o=n$, $p=l-1$ and where $a_n$ is equal to $\frac{H_n-a}{H_n}$, for some constant $a$. 
Compare the potentials of the strategy profiles which do not use edge $i$ for $i=0,\ldots, i=2\cdot l$ to the potential of $P$ (the strategy profile minimizing $\Phi$) that does not include the $p$-th edge. Note that after canceling the coefficients on both sides, the coefficient in front of $a_n$ is a sum of a constant number of terms converging to $0$ for $n$ tending to infinity, so these terms can be neglected. The potential of the strategy profiles which do not use edge $i$ for $\frac{n}{2}>i>2\cdot l$ is increasing when $i$ is increasing and decreasing towards $n$. 
We solved the resulting system of linear equations and obtained a lower bound for POPoA converging to $1.99992$ for $l=1000$ and $n$ tending to infinity. Thus, we have the following proposition. 

\begin{proposition}
There are games that have POPoA $1.99992$.   
\end{proposition}

We leave it as an open problem to analyze the convergence of the POPoA of the above construction, and conjecture that it converges to two.
%
%While the constant obtained above is close enough to $2$ for all practical reasons, unfortunately we were not able to prove that the solution $a_p$ is converging to $0$ while $l$ tends to infinity, but we have a good reason to believe so. While minimizing $a_p$ in the linear program, it was smaller than any number we tried, though for smaller numbers bigger linear programs had to be solved. The only problem to obtain better results was a real time. Therefore we conjecture the following: 

\begin{conjecture}
There are games that have POPoA arbitrarily close to $2$.
\end{conjecture}

\section{Myopic Sequential Prices of Anarchy/Stability}

In this section we study the myopic sequential price of anarchy and the myopic sequential price of stability.

%In \cite{sequentialOriginal} authors introduced sequential price of anarchy to measure the efficiency of a resulting network if players arrive one by one. Formally speaking, we have a permutation of players $\pi$. They come one by one and at each step each player chooses the least expensive path for connecting to a target vertex. The ratio of the resulting network cost in the worst case and the optimum network cost is sequential price of anarchy. We introduce a sequential price of stability for the multicast game, which is the ratio of the resulting network cost in the best case and the optimum network cost. 

\subsection{Sequential price of anarchy in multicast game on rings}

\begin{lemma}
The myopic sequential price of anarchy is at most $2$ in the multicast games on rings.
\end{lemma}

\begin{proof}
Consider an optimal strategy profile and let $o$ be the edge that is not used.
Consider any permutation (order) $\pi$ of the players. If any player $\pi(i)$, $i<o$, decides to take a path containing edge $o$ for the first time then it means that $a_o\leq \displaystyle \sum_{l=0}^{l=i}a_l$ which is bounded by the cost of optimum.
Therefore, the whole cost of the ring is bounded by $2$ times the cost of optimum. 
\qed
\end{proof}

\begin{figure}[t]
  \centering
  \includegraphics[width=0.5\linewidth]{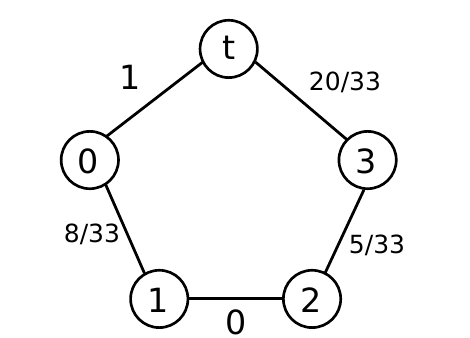}
  \caption{Lower bound example for the sequential price of anarchy}
  \label{sequentialPriceOfAnarchy}
\end{figure}
The presented upper bounds is tight, as shows the example in Fig. \ref{sequentialPriceOfAnarchy}, where $\pi=\{0,1,3,2\}$ results in myopic sequential price of anarchy equal to $2$.

\subsection{Myopic sequential price of stability in multicast game}

In the myopic sequential price of stability we consider the best permutation of players, with respect to the resulting network cost. In \cite{graphical-cost-sharing} authors prove that when the social knowledge network graph is directed acyclic then the price of anarchy is 
bounded by $4$ (Theorem 8). If we consider that in the social knowledge graph each incoming player knows all the previous players then the result can be directly translated into our setting, but we give a different (simpler than the proof of general result in \cite{graphical-cost-sharing}), proof for our setting: 

\begin{theorem}
The myopic sequential price of stability in multicast games on arbitrary graphs is at most $4$. 
\end{theorem}

\begin{proof}
Since there is a common target vertex $t$, any optimum strategy profile forms a Steiner tree $T$ on terminals $s_i$, $i=0,\ldots, n-1$ and $t$. 
Consider a permutation of the vertices that corresponds to a depth-first search of the tree $T$, and make it the identity permutation $(0,1,\ldots n-1)$. 
Let the players enter the game in this order, and make the myopic best responses. 
Denote by $B_i$ the cost of the edges that player $i$ uses alone in her strategy at the moment she enters the game, and let $S_i$ be the overall cost of player $i$ when she enters. 
Then the cost of the resulting network is $\displaystyle \sum_{i=0}^{i=n-1}B_i$. 
Since every player optimizes her cost when she enters the game, we have the following chain of inequalities: $S_i\leq d_T(s_i,s_{i-1}) + S_{i-1}-\frac{1}{2}B_{i-1}$, for $i=1,\ldots, n-1$, where $d_T(u,v)$ is the distance between nodes $u$ and $v$ using only the edges of the tree $T$. 
Each player $i$ has the following alternative strategy: first travel to the source (vertex) $s_{i-1}$ using the edges of $T$, and then follow the strategy of player $i-1$. 
Note that in this alternative strategy, player $i$ saves at least half of the cost of the edges that player $i-1$ takes alone when she enters the game. For the first player, we have the following inequality $S_0\leq d_T(s_0,t)$, because when she enters the game, one of the possible strategies is to take a direct path from $s_0$ to $t$ using only the edges of $T$. 
By summing up all inequalities given above, we get that $\frac{1}{2}\displaystyle \sum_{i=0}^{i=n-2}B_i+S_{n-1}\leq 2\cdot cost(T)$. Note that $S_{n-1}\geq \frac{1}{2}B_{n-1}$, which results into the upper bound of 4. 
\qed
\end{proof}

This upper bound is tight, as the example (Theorem 5) from \cite{graphical-cost-sharing} shows, ratio in the lower bound example is arbitrarily close to $4$.

\begin{proposition}
There is a multicast game with the myopic sequential price of anarchy arbitrarily close to $4$.
\end{proposition} 

\subsection{Myopic sequential price of stability on rings}

In this section we consider the myopic sequential price of stability of the multicast games on rings. 
The example from Fig. \ref{4/3 lower bound} shows that it can be as high as $\frac{4}{3}$. We prove the following upper bound.

\begin{theorem}
\label{sequentialPoS}
The myopic sequential price of stability in the multicast games on rings is at most $\frac{26}{19}$.
\end{theorem} 

\begin{proof}
Assume that the optimum strategy profile does not include the edge of cost $a_o$, and without loss of generality $\displaystyle \sum_{i=0}^{i=o-1}a_i\geq \displaystyle \sum_{i=o+1}^{i=n}a_i$.
Consider the permutation $\pi = \{n-1,\ldots, o, 0,1,\ldots,o-1\}$. First $n-o$ players clearly take the right path, by our assumption. Consider the remaining players. If there is no player which, upon arrival, prefers the right path over the left path, then only edges of an optimum strategy profile are included into the resulting network which means that the myopic sequential price of stability is $1$. 
If the very first player $0$ prefers the right path, then all other players necessarily prefer the right path as well, and the resulting network consists of all edges except for that of weight $a_0$. But then $a_0$ is at least as large as $a_o$, resulting again the myopic  sequential price of stability equal to $1$. 
Suppose that there exists $i$ such that every player $l\leq i$ prefers to take the left path, and only the player (vertex) $i+1$ prefers to take the right path. This implies the following inequalities: 
\begin{equation}
\label{i-thplayer}
\displaystyle \sum_{k=0}^{k=i}\frac{a_k}{i-k+1}\leq \displaystyle \sum_{k=i+1}^{k=n}a_k \text{, and}
\end{equation}

\begin{equation}
\label{i+1-thplayer}
\displaystyle \sum_{k=i+2}^{k=o}a_k\leq \displaystyle \sum_{k=0}^{k=i+1}\frac{a_k}{i+2-k} \text{,}
\end{equation} 
where the first inequality (\ref{i-thplayer}) indicates that the $i$-th player prefers the left path, and the second inequality (\ref{i+1-thplayer}) indicates that the $i+1$-th player prefers the right. 
Our goal is to investigate the maximum possible cost $c$ of the resulting network, where $c = a_0+\cdots+a_i+a_{i+2}+\cdots+a_n$. 
Take the first inequality (\ref{i-thplayer}) with weight $\frac{2}{19}$, the second inequality (\ref{i+1-thplayer}) with weight $\frac{24}{19}$, and the normalization equation $a_0+\cdots+a_{o-1}+a_{o+1}+\cdots+a_n = 1$ with weight $\frac{26}{19}$. 
We obtain that the sum on the left hand side $s$ satisfies $c\leq s\leq \frac{26}{19}$, which gives that $c\leq \frac{26}{19}\approx 1.368$. 
\qed
\end{proof}

The permutation from the proof of Theorem~\ref{sequentialPoS} cannot be used to provide a better bound, as there exists an example of a game, where the permutation results in a network of cost $\frac{26}{19}$ times larger than the cost of optimum.
The example consists of $3$ players and edges have weights $\frac{6}{19}$, $\frac{10}{19}$, $\frac{3}{19}$ and $\frac{10}{19}$ in the counter-clockwise order. 
Players who come in the game according to the permutation $\{0,1,2,3\}$ take all edges except for the $3$-rd edge of weight $\frac{3}{19}$, resulting into a network of cost $\frac{26}{19}$, while the optimum network cost is $1$. 
Note that if players come according to the ``opposite'' permutation $(n-1,\ldots, 0)$, then the resulting network has the same cost as the optimum network. 
We have experimentally played with these two permutations, and for all inputs we tried, one of the two permutations resulted in networks of cost no more than the 4/3 of the optimum cost.
Actually, we have checked that there is no instance of at most 1000 players where the better of the two permutations fails in that respect.

\begin{conjecture}
The myopic sequential price of stability in the multicast game on rings is at most $\frac{4}{3}$.
\end{conjecture}

\section{Conclusions}

We have analyzed several solution concepts for the multicast network design games on rings, and demonstrated that they differ in terms of quality. 
Some of the derived bounds are not shown to be tight, and we leave it for future work to make them tight. 

We have also initiated the study of the myopic sequential price of stability, and analyzed it for the multicast network design game on a ring. 
It is certainly an interesting challenge to provide better bounds on this concept for general (not multicast) network design games.
%
%We think that the methods developed in this paper will help to construct improved lower and upper bounds for other settings as well. We leave the question of studying myopic sequential price of stability for general network design game as a matter of future research.  
\\
\\
\noindent\textbf{Acknowledgements.} This work has been supported by the
Swiss National Science Foundation (SNF) under the grant number
200021\_143323/1. We used the CGAL linear and quadratic programming solver \cite{cgal:fgsw-lqps-15a} for solving all linear programs described in this article. 
We thank anonymous reviewer for insightful comments and especially for pointing out the relation between myopic sequential price of stability and graphical multicast cost sharing games.

\bibliographystyle{llncs2e/splncs03}
\bibliography{NetworkDesignGames}

\begin{thebibliography}{10}
\providecommand{\url}[1]{\texttt{#1}}
\providecommand{\urlprefix}{URL }

\bibitem{CallControl}
Adamy, U., Amb{\"{u}}hl, C., Anand, R.S., Erlebach, T.: Call control in rings.
  Algorithmica  47(3),  217--238 (2007)

\bibitem{sequentialIsolation}
Angelucci, A., Bilo, V., Flammini, M., Moscardelli, L.: On the sequential price
  of anarchy of isolation games. J. Comb. Optim.  29(1),  165--181 (2015)

\bibitem{original}
Anshelevich, E., Dasgupta, A., Kleinberg, J.M., Tardos, {\'E}., Wexler, T.,
  Roughgarden, T.: The price of stability for network design with fair cost
  allocation. In: FOCS. pp. 295--304 (2004)

\bibitem{Asadpour+Saberi/2009}
Asadpour, A., Saberi, A.: On the inefficiency ratio of stable equilibria in
  congestion games. In: WINE. pp. 545--552 (2009)

\bibitem{computingLeaderonrings}
Attiya, H., Snir, M., Warmuth, M.K.: Computing on an anonymous ring. J. {ACM}
  35(4),  845--875 (1988)

\bibitem{lowerbounds}
Bil{\`o}, V., Caragiannis, I., Fanelli, A., Monaco, G.: Improved lower bounds
  on the price of stability of undirected network design games. Theory Comput.
  Syst.  52(4),  668--686 (2013)

\bibitem{graphical-cost-sharing}
Bil{\`{o}}, V., Fanelli, A., Flammini, M., Moscardelli, L.: When ignorance
  helps: Graphical multicast cost sharing games. Theor. Comput. Sci.  411(3),
  660--671 (2010), \url{http://dx.doi.org/10.1016/j.tcs.2009.10.007}

\bibitem{constant}
Bil{\`o}, V., Flammini, M., Moscardelli, L.: The price of stability for
  undirected broadcast network design with fair cost allocation is constant.
  In: FOCS. pp. 638--647 (2013)

\bibitem{sequentialImproved}
Charikar, M., Karloff, H.J., Mathieu, C., Naor, J., Saks, M.E.: Online
  multicast with egalitarian cost sharing. In: {SPAA} 2008: Proceedings of the
  20th Annual {ACM} Symposium on Parallelism in Algorithms and Architectures,
  Munich, Germany, June 14-16, 2008. pp. 70--76 (2008)

\bibitem{sequentialOriginal}
Chekuri, C., Chuzhoy, J., Lewin{-}Eytan, L., Naor, J., Orda, A.:
  Non-cooperative multicast and facility location games. In: Proceedings 7th
  {ACM} Conference on Electronic Commerce (EC-2006), Ann Arbor, Michigan, USA,
  June 11-15, 2006. pp. 72--81 (2006)

\bibitem{Matus}
Disser, Y., Feldmann, A.E., Klimm, M., Mihal{\'a}k, M.: Improving the
  ${H}_k$-bound on the price of stability in undirected shapley network design
  games. In: CIAC. pp. 158--169 (2013)

\bibitem{ring-design-journalversion}
Fanelli, A., Leniowski, D., Monaco, G., Sankowski, P.: The ring design game
  with fair cost allocation. Theor. Comput. Sci.  562,  90--100 (2015),
  \url{http://dx.doi.org/10.1016/j.tcs.2014.09.035}

\bibitem{Fiat+etal/2006}
Fiat, A., Kaplan, H., Levy, M., Olonetsky, S., Shabo, R.: On the price of
  stability for designing undirected networks with fair cost allocations. In:
  ICALP. pp. 608--618 (2006)

\bibitem{cgal:fgsw-lqps-15a}
Fischer, K., G{\"a}rtner, B., Sch{\"o}nherr, S., Wessendorp, F.: Linear and
  quadratic programming solver. In: {CGAL} User and Reference Manual. {CGAL
  Editorial Board}, {4.6} edn. (2015),
  \url{http://doc.cgal.org/4.6/Manual/packages.htmlPkgQPSolverSummary}

\bibitem{multicast}
Jian, L.: An upper bound on the price of stability for undirected shapley
  network design games. Information Processing Letters  109,  876--878 (2009)

\bibitem{popos}
Kawase, Y., Makino, K.: Nash equilibria with minimum potential in undirected
  broadcast games. Theor. Comput. Sci.  482,  33--47 (2013)

\bibitem{Lee+Ligett/2013}
Lee, E., Ligett, K.: Improved bounds on the price of stability in network cost
  sharing games. In: EC. pp. 607--620 (2013)

\bibitem{sequentialPoA}
Leme, R.P., Syrgkanis, V., Tardos, {\'{E}}.: The curse of simultaneity. In:
  Innovations in Theoretical Computer Science 2012, Cambridge, MA, USA, January
  8-10, 2012. pp. 60--67 (2012)

\bibitem{mfcs2014}
Mamageishvili, A., Mihal{\'{a}}k, M., Montemezzani, S.: An ${H}_{\frac{n}{2}}$
  upper bound on the price of stability of undirected network design games. In:
  Mathematical Foundations of Computer Science 2014 - 39th International
  Symposium, {MFCS} 2014, Budapest, Hungary, August 25-29, 2014. Proceedings,
  Part {II}. pp. 541--552 (2014)

\end{thebibliography}

\newpage
\appendix

\section{ Weights for Inequalities from the Proof of Theorem~\ref{thm:pos<=4/3}}

In this appendix we provide the multiplicative weights of the inequalities using a dual to a linear program that was solved to upper bound the price of stability in the multicast game on rings. 
The first inequality is the normalization inequality, therefore its weight is the upper bound on the price of stability. The next $k$ inequalities indicate that the first $k\leq 7$ players left of edge $e$ prefer to deviate, i.e., prefer to choose the right path instead of the left path, and the last inequality indicates that we have a Nash equilibrium, i.e., the last player considered in the best-response dynamics prefers to stick with the left path than to switch to the right path. 
The objective of the linear program is to minimize the sum of the edge costs without the edge that is not used by the Nash equilibrium achieved via the best response dynamics. 
The coefficients (weights) are as follows:

\begin{itemize}
\item $k=1$ (0: 4/3;  1: 10/9;  2: 2/9)
\item $k=2$ (0: 22/17;  1: 252/323;  2: 202/323;  3: 90/323)
\item $k=3$ (0: 29/23;  1: 2976/4025;  2: 1206/4025;  3: 2256/4025;  4: 1224/4025)
\item $k=4$ (0: 1.243533565;  1: 0.722076586;  2: 446160/1659763;  3: 0.268809463;  4: 0.528169383;  5: 0.329251827)
\item $k=5$ (0: 1.229596836;  1: 0.711037768;  2: 0.257115234;  3: 0.201170436;  4: 0.199302216;  5: 0.50797093;  6: 0.348431623)
\item $k=6$ (0: 1.217310111;  1: 0.702648246;  2: 0.250967669;  3: 0.189905238;  4: 0.168566505;  5: 0.179311025;  6: 0.494134279;  7: 0.362553601) 
\item $k=7$ (0: 1.206536915;  1: 0.69586637;  2: 0.247111078;  3: 0.184286036;  4: 0.157438535;  5: 0.148587957;  6: 0.165607593;  7: 0.484007846;  8: 0.373384452)
\end{itemize}

For $k>7$, we take only the first $7$ inequalities indicating that the first $7$ players prefer to take the right path than to stick to the left path. This is enough to prove an upper bound of $1.33081$ for the price of stability. In the following, we list the weights of the inequalities of the dual to our linear program (index $k:$ denotes the weight of the inequality to player $k$):
(0: 1.330802428;  1: 0.750587484;  2: 0.246845878;  3: 0.168106752;  4: 0.12615003;  5: 0.096800836;  6: 0.072578056;  7: 0.048719834).

\end{document}